\documentclass{amsart}

\pdfoutput=1 

\usepackage[pdftex]{graphicx}

\usepackage{bbm}

\usepackage[latin1]{inputenc}
\usepackage[pdftex]{epsfig}
\usepackage[pdftex]{color}
\usepackage{epstopdf}

\newtheorem{proposition}{Proposition}

\def\P{{\mathbb P}}
\def\E{{\mathbb E}}

\newcommand\ind[1]{\mathbbm{1}_{\left\{#1\right\}}}
                                           
\title{A Flow-aware MAC Protocol  for a Passive Optical Metropolitan Area Network }

\author[Ph. Robert]{Philippe  Robert}

\address[Ph. Robert, J. Roberts]{INRIA Paris --- Rocquencourt,  Domaine de Voluceau, 78153 Le Chesnay, France.}

\email{Philippe.Robert@inria.fr}
\urladdr{http://www-rocq.inria.fr/\string~robert}

\author[J. Roberts]{James  Roberts}
\email{James.Roberts@inria.fr}

\begin{document}
\maketitle

\begin{abstract}
The paper introduces an original MAC protocol for a passive optical metropolitan area network using time-domain wavelength interleaved  networking (TWIN)
. Optical channels are shared under the distributed control of destinations using a packet-based polling algorithm. This MAC is inspired more by EPON dynamic bandwidth allocation than the slotted, GPON-like access control generally envisaged for TWIN. Management of source-destination traffic streams is flow-aware with the size of allocated time slices being proportional to the number of active flows. This emulates a network-wide, distributed fair queuing scheduler, bringing the well-known implicit service differentiation and robustness advantages of this mechanism to the metro area network.  The paper presents a comprehensive performance evaluation based on analytical modelling supported by simulations. The proposed MAC is shown to have excellent performance in terms of both traffic capacity and packet latency.

\end{abstract}

\section{Introduction}

Among recent proposals for realizing a metropolitan area network (MAN) using optical technology, time-domain wavelength interleaved networking (TWIN) is a particularly attractive alternative, allowing cost effective, energy efficient communication using currently available technology \cite{Ross03,Saniee09}. TWIN uses wavelength selective optical cross-connects (OXCs) to create multipoint-to-point lightpaths in the form of trees, each connecting source routers to  a particular destination router. The lightpath wavelength in effect constitutes the destination router's address. The OXCs are programmed to passively direct all incoming light on a given wavelength to a particular outgoing fibre, bringing the signals progressively closer to the destination. 

Figure \ref{fig:twin} depicts the tree giving access to router $R_1$. Any source can send signals to $R_1$ simply by emitting them in the form of light bursts on the corresponding wavelength. Sources are equipped with one or more fast-tunable transmitters able to send bursts successively on all wavelengths.
It is important to realize that any bursts that are timed not to collide at the destination, cannot collide anywhere else in the network.  The drawing on the right is thus the logical equivalent of the network on the left. Of course, every arc in this graph would bear one lightpath in each direction but these are not represented for the sake of clarity. 

\begin{figure}[th]
\centering
\resizebox{8cm}{!}{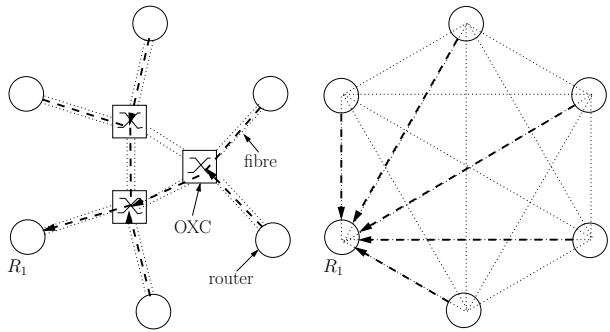}
\caption{A six-node TWIN MAN: fibre infrastructure (left) is used to create destination rooted trees (right); dashed lines represent the multipoint to point lightpath serving router $R_1$. }\label{fig:twin}
\end{figure}

In this paper, we propose an original distributed MAC protocol to manage bandwidth sharing on the network lightpaths. As in \cite{Saniee09}, we suppose each destination independently orchestrates optical burst transmissions from the sources by allocating grants for bursts that are timed not to collide.  However, each source receives grants from several destinations and may not be able to fulfill all because it has only a limited number of transmitters. It will then have to partially or totally ignore one or more grants, leading to a loss of capacity.. 

Unlike the MAC envisaged for TWIN  in \cite{Ross03,Saniee09}, we do not impose a rigid time frame structure but assume rather that destinations issue grants for arbitrarily defined intervals specified by their start time and duration. This allows a flexible dynamic bandwidth allocation algorithm more akin to that of an EPON access network \cite{epontutorial} than that of the alternative GPON frame-based standards \cite{G9843}. Grant sizes, as for EPON, could be determined according to a variety of service policies. We propose here to apply a particular policy that realizes a form of flow-aware networking \cite{OR05}.

We assume flows can be reliably identified `on the fly' from packet header fields and that routers implement per-flow fair scheduling for each source-destination traffic relation. In essence (details are given later), each source periodically reports to the destination the current number of active flows, i.e., the number of flows currently holding at least one packet in the buffer. The destination issues grants to the source in return allowing it to send a burst including a `quantum' of bytes for each reported flow. The quantum size would typically be equivalent to one or several packets. All sources are allocated grants with the same frequency so that this service policy essentially realizes network-wide, per-flow fair sharing of lightpath bandwidth.

 As discussed in \cite{OR05}, per-flow fair sharing has two principal advantages. It realizes implicit service differentiation since streaming and conversational flows typically have a rate less than the fair rate and therefore experience low packet latency. It allows potentially high rate elastic flows to efficiently exploit residual bandwidth without any requirement for end systems to implement a particular ``TCP friendly'' congestion control algorithm.

The paper first presents our original MAC protocol that combines the advantages of TWIN passive optical networking with the simplicity and efficiency of flow-aware networking. We then proceed to the performance evaluation of a single lightpath tree, using analytical modelling backed up by simulation. The analysis proves that traffic capacity is optimal and demonstrates the network's excellent performance in terms of both packet latency and realized flow throughput.  The performance of an entire network is then evaluated, taking into account the loss of capacity due to transmitter blocking. We evaluate the amount of lost capacity as a function of the number of tunable transmitters equipping each source router.

 \section{A flow-aware MAC}
 

The envisaged MAC protocol relies on each destination router independently allocating grants to its source routers that are timed not to collide. Each source arbitrates between overlapping grants from different destinations when their number exceeds its transmission capability.   Grants take the form of time slices on the appropriate wavelength channel specified by a start time and a duration.  

%
 
 \subsection{Signalling reports and grants}
We assume sources report their current buffer contents to respective destinations using constant length messages. Time to send these reports is included in the grants attributed by the destination. A report is always tagged to the end of any data transmission and reports are also sent in isolation when a source has no data to send. 

Each destination continually emits grants to its source nodes realizing a kind of polling system designed to ensure new arrivals at the source are reported as soon as possible. The polling scheme is inspired by the EPON dynamic bandwidth allocation algorithm described in \cite{AFRR10b}. Unlike the EPON, however, we assume packets can be fragmented at will to fully utilize an assigned grant. 

While reports are signalled in-band, using the network lightpaths, it appears necessary to use out-of-band signalling over external media to communicate grants. A possible in-band signalling implementation for sending grants from destination $j$ to source $i$ would be to include them in the bursts previously granted to source $j$ by destination $i$. 
Unfortunately, this appears to lead to unavoidable deadlocks where both $i$ and $j$ have grants to send but neither has a scheduled burst in which they can be sent. This occurs because, sometimes, both $i$ and $j$ only generate a new grant {\em after} they have already fulfilled their latest grant in the opposite direction.  We therefore suppose grants are sent over some unspecified other network. For the present work, we characterize this simply by an assumed maximum grant transmission time.

\subsection{Synchronization}
All nodes must be carefully synchronized in real time to ensure precise transmission schedules are realizable in practice. This is possible using a regular exchange of time stamps, as performed in EPON  \cite{epontutorial}. The following procedure also measures the round trip propagation time between each pair of nodes. 

Referring to Figure \ref{fig:synchro}, node $j$ is the destination corresponding to some wavelength $\lambda_j$. This node emits a report to node $i$ on the appropriate wavelength, $\lambda_i$ say, at $j$'s local time $t_1$. Node $j$ writes time $t_1$ as a time stamp in the message. On receipt, the local clock of node $i$ is set to $t_1$. When node $i$ next sends a report to $j$ on $\lambda_j$ it time stamps the message with its local time of emission $t_2$. Node $j$ calculates the round trip time as shown in the figure caption. The local clock at $i$ is slow by unknown propagation time $\delta_{ji}$ with respect to clock at $j$. This shift is automatically taken into account by the algorithm described next that only needs to know the round trip time $\textsc{rtt}_{ij}$. In a MAN with $R$ routers, each one must maintain a separate local clock for each wavelength, one as destination and $R-1$ as source.

\begin{figure}[th]
\centering
\resizebox{7cm}{!}{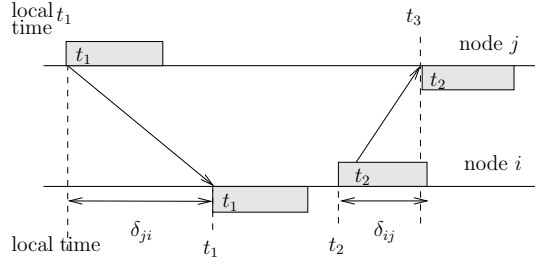}
\caption{Measuring round trip times: $\textsc{rtt}_{ij} = \delta_{ij}+\delta_{ji} =(t_3 - t_1)-(t_2-t_1)=t_3-t_2$. }
\label{fig:synchro}
\end{figure}

\subsection{Grant recursions}
In computing grant epochs (i.e., the moment the grant is formulated using available reports; it is sent as soon as possible after that), a destination node $j$ must take account of the delay incurred before the burst arrives at the destination as well as the guard time necessary when a transmitter switches between wavelengths. A guard time of 1.5$\mu$s is standard for EPON. The grant transit delay includes propagation, transmission and any waiting time. We assume in the following that the delay from destination node $j$ to source node $i$ is bounded with some suitably high probability by $\delta_{ji} + \tau$. Time $\tau$ is an increment to the lightpath propagation time from $j$ to $i$. Its value depends on how grant signalling is actually performed. For the present work, we assume $\tau$ is given. 

The size of the grant is computed based on reported queue contents, on applying a particular service policy. The flow-aware service policy considered here is described later in Section \ref{sec:allocation}. The grant must also include time to send the report and the guard time before the channel can be used by another source. We denote the sum of these times by $\Delta_R$. 

The process of grants emitted by destination $j$ to all other nodes is specified by the functions $g(n)$, $s(n)$ and $d(n)$ defined as follows. The $n^{th}$ grant sent to some source by destination $j$ is formulated at nominal time $g(n)$ and instructs the source to transmit for duration $d(n)$ starting at  {\it source local time} $s(n)$. Assume the $(n+1)^{th}$ grant is issued to  source $i$.  Epochs $g$ and $s$ are calculated recursively as specified in Proposition \ref{prop:recursion}. 

\begin{proposition}
\label{prop:recursion}
The following recursions define a schedule that is feasible and ensures the optical channel is fully utilized:  
 \begin{eqnarray}
  g(n+1) &= &g(n) + d(n) + \Delta_R, \label{eq:g-update}\\
 s(n+1) &= &g(n+1) + \Delta_O  - \textsc{rtt}_{ij}, \label{eq:s-update}
\end{eqnarray}
where $\Delta_O$  is an offset satisfying $\Delta_O \geq  \max_i(\textsc{rtt}_{ij})+\tau$.
\end{proposition}
 
 \begin{proof}
 Feasibility requires $g(n+1)   + \tau \leq s(n+1)$, the grant must arrive at the source before the scheduled start time accounting for maximal delay $\tau$.
This follows from (\ref{eq:s-update}) and $\Delta_O \geq  \textsc{rtt}_{ij} +\tau$ (recall that $s(n+1)$ is the start time measured at the source clock).
The $n^{th}$ grant schedules a burst whose leading edge arrives at the destination at (destination clock) time $s(n) + \textsc{rtt}_{ij} = g(n)+\Delta_O$. The channel is free for another burst to arrive  $d(n)+\Delta_R$ seconds later.  By (\ref{eq:g-update}), this time coincides with the arrival time of the next burst $g(n+1)+\Delta_O$, demonstrating that the channel is indeed fully utilized. \end{proof}

The sequence in which source nodes are attributed grants can be arbitrary. However, the phenomenon of transmitter blocking explained next can lead to poor performance when the sequence is deterministic. We therefore apply a randomized scheme in the evaluations presented below: the source to receive the $(n+1)^{th}$ grant is chosen uniformly at random from all sources except the one receiving grant $n$.

\subsection{Transmitter blocking}
\label{sec:blocking}
A source node receives grants from different destinations and these can overlap. If the number of overlaps is greater than the number of source transmitters, one or more grants cannot be fully satisfied. We assume the node fulfills grants in their arrival order. When one satisfied grant ends and some other unsatisfied grant has not entirely expired, the transmitter is retuned to the corresponding wavelength for the remaining grant interval. As with the slotted algorithm described in \cite{Saniee09}, transmitter blocking leads to lost capacity, as analyzed in Section \ref{sec:blockingperf} below.

 
\subsection{Resource allocation}
\label{sec:allocation}

One could implement a variety of different algorithms to determine the grant durations $d(n)$ in Proposition \ref{prop:recursion}. The present proposal derives from previously published arguments that resource sharing in networks should be flow-aware and that performance requirements can be satisfied by two mechanisms: per-flow fair scheduling in router output queues and an overload control intended to maintain efficiency when demand approaches or exceeds capacity \cite{OR05}. 

We suppose flows can be reliably identified and, to simplify the presentation, that they clearly fall into one of two categories: {\it backlogged} flows that have no exogenous rate limit beyond the considered MAN and always maintain a backlog of packets in the buffer, and  {\it non-backlogged} flows that are limited in rate elsewhere such that, with fair scheduling, they never have more than one packet in the buffer. 

The grant is intended to cater for all queued packets of non-backlogged flows and one ``quantum'' of bytes for each backlogged flow. This service policy can be realized using a priority fair queuing scheduler like that described in  \cite{KOR04} or  \cite{KOR05} that automatically distinguishes the two types of flow. 

Reports indicate the amount (in seconds of transmission time at the optical channel rate) of non-backlogged traffic arrived since the last report was sent together with the current number of backlogged flows. To account for grants that are wholly or partially unfulfilled due to transmitter blocking (or to grants arriving after their nominal start time), we introduce the notion of ``deficit''. The deficit is equal to the incremental amount of grant time, nominally available for emission since the last report was sent, that has not in fact been used. It is computed by the source and added to the count of new non-backlogged arrivals in the next report to be sent. 

When a grant is emitted by the destination, it calculates the allocation based on all reports received since the last grant was issued, i.e., the sum of non-backlogged arrivals and deficits plus the time to send one quantum of each backlogged flow reported in the last received report. These reporting and granting mechanisms ensure all buffered traffic is eventually served.

\subsection{Burst formation}
\label{sec:formation}

\begin{figure}[th]
\centering
\resizebox{9cm}{!}{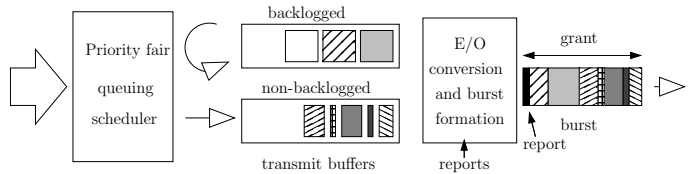}
\caption{Burst formation: packets and fragments are assembled to fulfill each grant. }
\label{fig:burstformation}
\end{figure}

Figure \ref{fig:burstformation} illustrates the way packets arriving to the router for a given destination are formed into bursts. It is assumed in the drawing that there are no deficits to be taken into account and the grant has not suffered blocking. We assume packets can be freely fragmented in forming the optical bursts. They are reassembled after optical-electrical conversion on receipt of all fragments.

When a grant is fulfilled by a source at the designated start time, the buffer contents will typically not be the same as that reported because of the scheduling delay. The grant is used first to send packets of non-backlogged flows, including any that may have arrived since the report was issued. The remainder is used to send as many quanta of backlogged flows as possible, serving flows in round robin sequence, starting where service of the last grant ended.

Often, the grant is not sufficient to serve one quantum of each backlogged flow (as in the figure). Occasionally, because some backlogged flow has just ended, the grant may be too large and a small amount of capacity will be wasted. A more significant loss of capacity occurs because of transmitter blocking, as discussed in Section \ref{sec:blocking}.

%
%

\subsection{Overload control}
While per-flow scheduling is feasible for loads up to around 90\% (see results of Section \ref{sec:treesim} below), the number of flows will increase unboundedly if load should approach or exceed 100\% of wavelength capacity. It is necessary therefore to implement some form of overload control, both to preserve performance and to ensure the number of backlogged flows to be scheduled remains relatively small (less than 100, say). Overload would typically be manifested first by some source locally observing an inordinate number of backlogged flows in progress. That source would activate a load reduction mechanism (e.g., discarding the packets of a certain set of flows). Any other source observing overload would behave similarly leading, eventually, to an overall load that is manageable (not more than 90\%, say). Discussing precise details of this mechanism is beyond present scope.

 \section{Performance of an isolated destination tree}
 \label{sec:onetree}
 
We first consider the performance of an isolated destination tree of capacity $C$ (bits/sec), ignoring the impact of blocking due to transmitter contention at the source.  Let the number of sources be $S$ and assume demand due to source $i$ is $\rho_i C$ where $\rho_i$ is the channel load equal to the product of flow arrival rate (flows/sec) and mean flow size (bits) divided by $C$. We distinguish loads $\rho_i^B$ due backlogged flows  and $\rho_i^N$ due to non-backlogged flows  with $\rho_i^B + \rho_i^N = \rho_i$. Overall load is $\rho = \sum \rho_i$.

%



 \subsection{Traffic capacity}
 
Traffic capacity is defined as the limiting demand beyond which queues would grow indefinitely.
With the considered flow-aware MAC, it seems intuitively clear that the size of the reporting and guard time overhead $\Delta_R$ does not impact capacity since grants become larger as load increases leading to low {\it relative} overhead. The following theorem states this result for a system whose traffic consists of backlogged flows alone.

\begin{proposition}
\label{prop:capacity}
Assuming all flows are backlogged, arrive as a Poisson process and have a general size distribution with finite second moment, the  considered destination tree with per flow service is stable if and only if $\rho < 1$.
\end{proposition} 
That $\rho<1$ is necessary is obvious. The proof of sufficiency is outlined in the appendix. 
 As for any stable equitable polling system, the expected cycle time between successive visits to the same source is $ S \Delta_R / (1-\rho)$. This expression is valid for quite general traffic characteristics. 
 
 \subsection{A processor sharing approximation}
 \label{sec:ps-approx}
To estimate throughput performance and the distribution of the number of flows in progress, we consider the following limiting system. In each grant, the destination attributes a quantum $q$ of service to each flow in progress and spends an overhead of $\Delta_R = xq/S$ before moving to the next source. We consider the limit where $q \rightarrow 0$. This corresponds to a processor sharing system with a permanent customer having a relative service requirement  of $x$. First, assume all flows are backlogged and flow sizes have an exponential distribution. 

\subsubsection{Stationary distributions}
\label{sec:stationary}
Let $N_i(t)$ be the number of source $i$ flows in progress at time $t$ and let $M(t)=\sum N_i(t)$ be the overall flow population.  Process $(N(t))$ can be considered as a network of processor sharing queues \cite{Serfoso,BP02} where the service rate $\phi_i(n)$ of a flow at queue $i$ when $N_i(t) = n_i$ for $i = 1,\ldots,S$, is:
$$ \phi_i(n) = \frac{n_i}{m + x}, $$
where $m=\sum n_i$. It is easy to verify that these service rates are balanced, i.e., for all $i,j$,
$$\phi_i(n)\phi_j(n-e_i) = \phi_i(n-e_j)\phi_j(n). $$
We deduce that $(N(t))$ is a Whittle network \cite{BP02}  and that the system is therefore stable iff $\rho<1$ with stationary distribution
\begin{equation}
 \pi(n) = (m+x)_m \prod \frac{\rho_i^{n_i}}{n_i!} \; (1-\rho)^{(1+x)},
\end{equation}
where notation $(y)_r$ denotes $y(y-1)\dots(y-r+1)$.

The stationary distribution of $(M(t))$ is
\begin{equation}
\omega(m) = (m+x)_m\frac{\rho^m}{m!} \; (1-\rho)^{(1+x)},
\end{equation}
and the marginal distribution of $(N_i(t))$ is,
\begin{equation}
\pi_i(n_i) = (n_i+x)_{n_i} \frac{\tilde\rho_i^{n_i}}{n_i!} \; (1-\tilde\rho_i)^{(1+x)},
\label{eq:pi-i}
\end{equation}
where $\tilde\rho_i = \rho_i/(1 - \rho +\rho_i)$. The expected number of flows in progress at source $i$ is $\E(N_i(t))=\rho_i(1+x)/(1-\rho)$.

\subsubsection{Response times and throughput}


From \cite[Proposition 5]{BP02}, since $(N(t))$ is a Whittle network, the expected response time of a flow of size $s$ is proportional to $s$. From the same reference, the constant of proportionality for a source $i$ flow is $\E(N_i(t))/(\rho_i C)$.  We deduce the expected response time $R(s)$ of any flow of size $s$,
$$
R(s) = \frac{s}{C} \frac{(1+x)}{(1-\rho)}.
$$
Defining flow throughput $\gamma$ as the ratio $s/R(s)$, we have
\begin{equation}
\gamma  = (1-\rho)C/(1+x).
\label{eq:gamma}
\end{equation}

\subsubsection{Insensitivity}
Since the service rates of $(N(t))$ are balanced, all the results in Section \ref{sec:ps-approx} are true for general flow size distributions \cite{BP02}. They are true also if flows arrive in ``sessions''  and session arrivals are Poisson \cite{BMPV06}. Each session is a succession of flows separated by intervals between the end of one flow and the start of the next. The flow sizes and interval lengths can have general distributions and be correlated, and the distribution of the number of flows in a session can be general. These variables cannot, however, depend dynamically on the system state.  


\subsubsection{Accounting for non-backlogged traffic}
The above model ignores the impact of flows that are not backlogged. These flows can be incorporated approximately as follows. 

We assume non-backlogged flows, being handled with priority (cf. Section \ref{sec:allocation} above), simply reduce available capacity and demand: $C \leftarrow C(1-\sum \rho_i^N)$ and $\rho_i \leftarrow \rho_i^B$. In particular, flow throughput would then still be given by (\ref{eq:gamma}).


\subsection{Simulations}
\label{sec:treesim}
 
We have developed an {\it ad hoc} simulator in C. Packets of non-backlogged flows are considered as a composite stream handled with priority. Arrivals are modelled as a variable rate Poisson process. The rate depends on the number of non-backlogged flows in progress that we assume varies like the population of an M/M/$\infty$ system.  Backlogged flows arrive as a Poisson process and have a size drawn from an exponential distribution. We assume sufficient packets of these flows are always present to fulfill received grants until the flow has ended. In these simulations and those reported in Section \ref{sec:netsimulation} below, we have assumed all grants arrive on time.  We assume the grant delay tolerance is $\tau = 1$ms. The following parameters characterize the considered system configuration:

\vspace{2mm}
\begin{center}
\begin{tabular}{lr}
\hline
number of source nodes ($R-1$), & 10\\
channel capacity ($C$), & 1Gb/s\\
constant packet size, & 1KB\\
report + guard time ($\Delta_R$), & 2$\mu$s\\
rate of non-backlogged flows, & 2Mb/s\\
mean duration of non-backlogged flows, & 30s\\
mean size of backlogged flows, & 10MB\\
service quantum ($q$), & 1KB\\
\hline
\end{tabular}
\end{center}
\vspace{2mm}

\noindent The round trip propagation time between each source and the destination is drawn at random between .02 and 1ms, corresponding to distances of up to 100km.  

\begin{figure}[th]
\centering
\includegraphics[scale=.7]{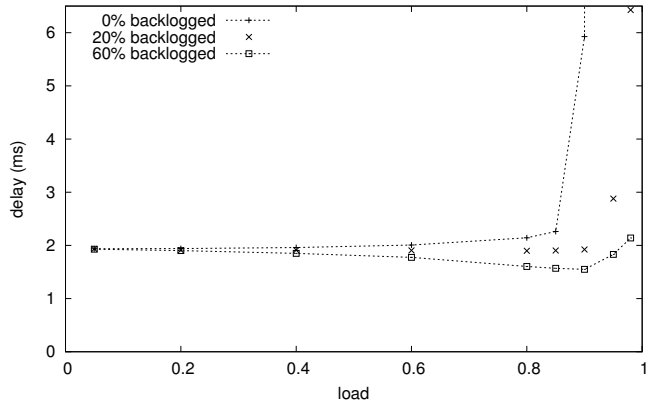}
\caption{Mean non-backlogged flow packet delay against load for traffic mixes with 0\%, 20\% and 60\% percent of backlogged flows - one destination tree}\label{fig:delay}
\end{figure}

Figure \ref{fig:delay} shows how the mean delay of  non-backlogged flow  packets varies with load. Results are presented for three proportions of backlogged flow traffic: 0\%, 20\% and 60\%. Note first that delays are small until load gets very close to capacity, except when traffic is 100\% non-backlogged. This case is not really representative, however, since when load is too high, even low rate flows actually become momentarily backlogged and would not in practice be given priority. 

Delay at low load is dominated by the report-grant exchange necessary to account for a new arrival. It is the same for all sources and is equal to the offset $\Delta_O$.  The service policy described in Section \ref{sec:allocation} is such that, when the proportion of backlogged flow  traffic is significant, delay first decreases with load. This is because newly arriving non-backlogged flow packets effectively ``steal'' the grant accorded earlier to backlogged flows. 

\begin{figure}[th]

\includegraphics[scale=.72]{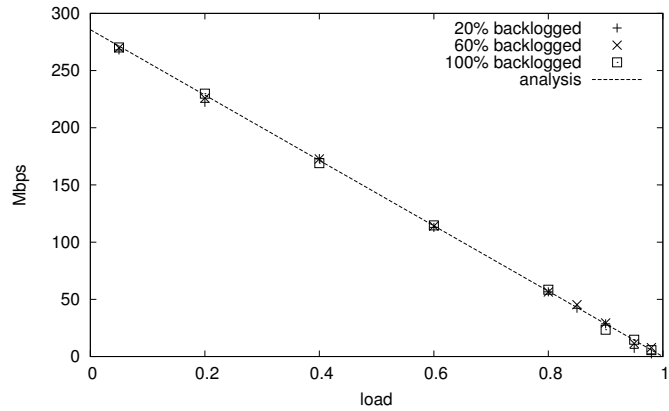}
\caption{Mean backlogged flow throughput against load for traffic mixes with 20\%, 60\% and \%100 percent of backlogged flows - one destination tree}
\label{fig:throughput}
\end{figure}

Figure \ref{fig:throughput} shows the mean throughput of the elastic flows. The figure shows simulation results for 20\%, 60\% and 100\% of elastic traffic. The figure also displays approximation (\ref{eq:gamma}), these results confirming its accuracy. As predicted by the analysis, throughput is greater when the quantum is increased. A quantum of 10KB yields a maximum throughput (at load zero) of 800Mb/s (results not shown).


Figure \ref{fig:histogram} shows the distribution of the number of flows at each node when all traffic is elastic and overall load is 90\%. Two cases are represented: one with a 1KB quantum ($x=2.5$) and one with 10KB ($x=.25$). The figure shows simulation results together with the analytical estimation (\ref{eq:pi-i}), confirming the accuracy of the latter. It is important to observe that the number of flows to be taken into account is relatively small, confirming that per-flow scheduling is scalable (see \cite{KMOR05a}). 


\begin{figure}[th]
\centering
\includegraphics[scale=.72  ]{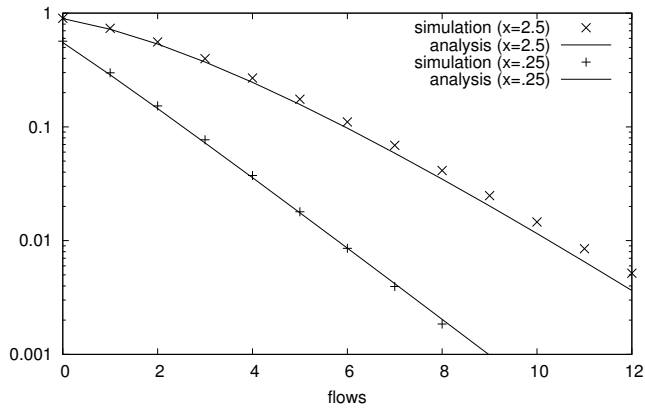}
\caption{Distribution of number of active flows in a given node }\label{fig:histogram}
\end{figure}

\section{Transmitter blocking}
\label{sec:blockingperf}

\begin{figure}[th]
\centering
\resizebox{8cm}{!}{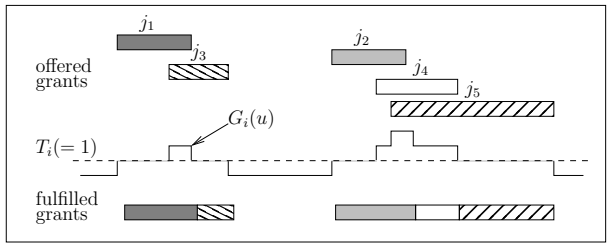}
\caption{Grants to source $S$ partially overlap leading to blocking. }\label{fig:blocking}
\end{figure}

Since a source is typically equipped with a number of tunable transmitters less than the number of destinations, actual performance is somewhat worse than presented in Section \ref{sec:onetree} because of transmitter blocking.
Figure \ref{fig:blocking} illustrates a situation where grants issued to source node $i$ by five destinations, $j_1$ to $j_5$, partially overlap leading to blocking: the number of grants $G_i(u)$ issued for intervals that include time $u$ momentarily exceeds $T_i$, the number  of transmitters. We assume grants are served in arrival order leading to the source activity depicted at the bottom of the figure.

\subsection{Impact on performance}
To evaluate the impact of blocking we assume for the sake of simplicity that traffic in the network of $R$ nodes is symmetric:  the processes of grants issued to source $i$ by the different destinations are then statistically identical. The processes are not independent, however, since blocked portions of grants are not lost but contribute to the size of subsequent grants. We nevertheless assume this is the case and that the only impact of blocking is to increase the intensity of each grant process, as discussed below. The independence assumption is intuitively more reasonable as the number of nodes increases. We successively consider heavy traffic and light traffic approximations.

\subsubsection{Heavy traffic}
\label{sec:heavytraffic}
This approximation is useful in estimating the system traffic capacity. We assume the impact of the report message and guard time overhead is negligible in this regime.
Let the proportion of blocked grant time when $T_i=t$ be $B_t(\rho)$ where $\rho C/(R-1)$ is the demand from the source to each destination. Assuming stability, the probability a given destination issues a grant including an arbitrary instant $u$ is then $\rho'/(R-1)= \rho/(R-1)/(1-B_t(\rho))$.  Let $g_n = \P(G_i(u)=n)$ be the stationary distribution of the number of grants encompassing $u$. By the independence assumption, we have
\begin{equation}
g_n = {R-1 \choose n}\left(\frac {\rho'}{R-1}\right)^n \left(1-\frac {\rho'}{R-1}\right)^{R-1-n}.
\label{eq:g-dist}
\end{equation}
The proportion of blocked grant time, which by the symmetry assumption is the same for all destinations, is then,
\begin{equation}
B_t(\rho) = \sum_{n>t} (n - t) g_n/\sum_{n>0}i g_n.
\label{eq:blocking}
\end{equation}
The $t$ transmitters are fully used when $\rho' \rightarrow t$. Setting $\rho'$ to $t$, we deduce from (\ref{eq:g-dist}) and (\ref{eq:blocking}) the maximum allowed load $\rho^*$ and the corresponding fractional loss of capacity $B_t(\rho^*)$ for each destination tree.

 \begin{figure}[th]
\includegraphics[scale=.65]{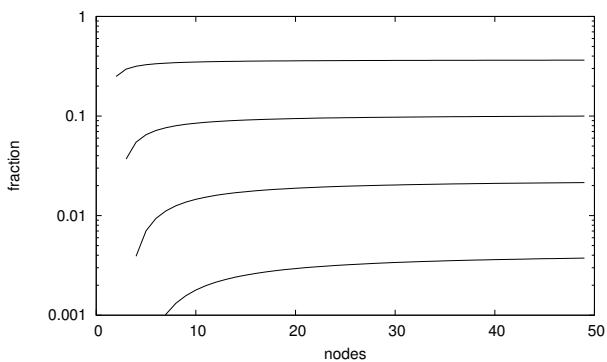}
\caption{Fractional loss of capacity for, from top to bottom, 1, 2, 3 and 4 tunable transmitters.}
\label{fig:prob-block}
\end{figure}

Figure \ref{fig:prob-block} plots the loss fraction $B_t(\rho^*)$ as a function of the number of network nodes, $R$, and the number of transmitters each one has, $t$. This grows rapidly to a limit as $R$ increases. In particular, we find $B_1(\rho) \rightarrow  e^{-1} \approx 0.37$, the same value derived for TWIN in \cite{Saniee09}. This is a significant loss in capacity that can be considerably reduced at the cost of additional transmitters: $B_2(\rho) \rightarrow 0.10$, $B_3(\rho) \rightarrow 0.02$. The `more transmitters/less fibre' tradeoff  may or may not be favourable depending on the economics of a particular network configuration. 

\subsubsection{Light traffic}
Throughput and latency in light traffic are dominated by the impact of the report and switch overhead. In this regime, most of the time sources have no data to send and spend their time emitting reports. A newly arrived flow competes for throughput with these reports. If the time to send a quantum of data is $q$ and the time for a report and ensuing guard time is $\Delta_R=xq/S$, throughput at near zero load in the absence of blocking would be $C/(1+x)$. A quantum of data can be partially blocked by a report to be sent to another destination but the amount of lost capacity turns out to be quite small. Ignoring this and assuming linearity between light and heavy traffic, we deduce the following approximation for throughput:
\begin{equation}
\gamma \approx \left(1-\frac{\rho}{1-B_t}\right) \cdot \frac{C}{1+x}.
\label{eq:gamma-blocking}
\end{equation}
  
\subsection{Simulations}
\label{sec:netsimulation}

We have simulated a symmetrical network with 10 sources and 10 destinations. 
The other simulation parameters are as in Section \ref{sec:treesim}. 

\begin{figure}[th]
\includegraphics[scale=.72]{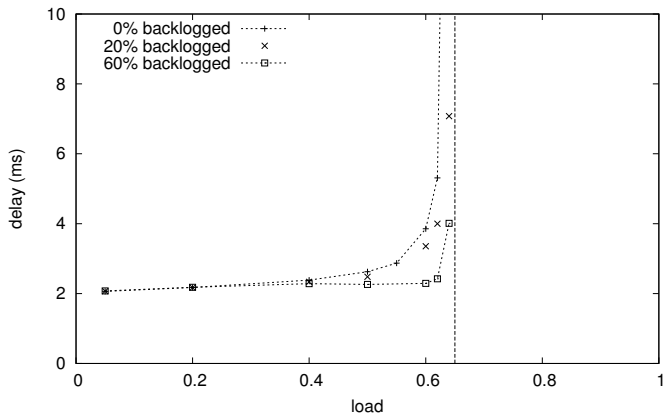}
\caption{Mean non-backlogged flow packet delay against load for traffic mixes with 0\%, 20\% and \%60 percent of backlogged flows - network of 10 destination trees, 1 transmitter per source }\label{fig:netdelay}
\end{figure}

Figure \ref{fig:netdelay} depicts the mean packet delay of non-backlogged flows for a traffic mix with 0\%, 20\% and 60\% of load from backlogged flows. The vertical line corresponds to the limiting load computed as in Section \ref{sec:heavytraffic} for 1 transmitter and 10 nodes ($\approx .65$). Note that delay is very small until load approaches this limit, even for the unfavourable case with 0\% backlogged traffic. 

\begin{figure}[th]
\includegraphics[scale=.72]{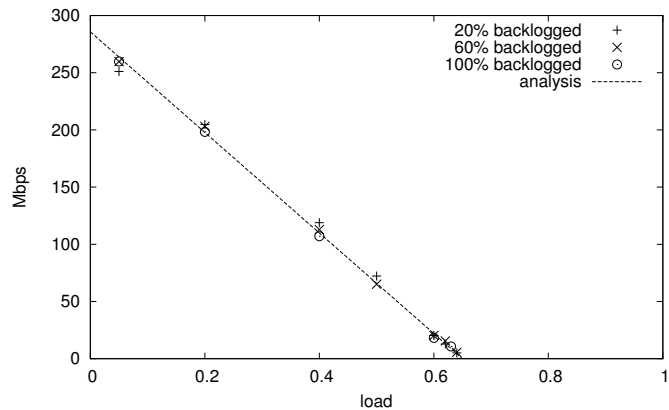}
\caption{Mean backlogged flow throughput against load for traffic mixes with 20\%, 60\% and \%100 percent of backlogged flows - network of 10 destination trees, 1 transmitter per source }\label{fig:netthru}
\end{figure}

Figure  \ref{fig:netthru} confirms that throughput depends linearly on load and remains insensitive to the traffic mix. Approximation (\ref{eq:gamma-blocking}) is accurate in this case.


 \section{Conclusion}
We have presented a MAN architecture based on passive optical network technology inspired by the TWIN proposal from Bell Labs \cite{Ross03,Saniee09}. Our principal contribution is to propose an original distributed flow-aware MAC protocol. Variable size bursts are scheduled by destination nodes to accommodate one quantum of traffic from each flow currently holding packets in the source buffer. This realizes network-wide per-flow fair sharing of each wavelength channel, which is arguably sufficient to realize the performance requirements of all types of flow. As in the original TWIN proposal, traffic capacity is reduced by the phenomenon of transmitter blocking when bursts, scheduled independently by different destinations, overlap at the source.

The performance of the MAC protocol has been evaluated by a combination of analysis and simulation. Latency of low rate, non-backlogged flows is confirmed to be very low, being dominated until demand approaches limit capacity by the time to exchange report and grant messages. Traffic capacity is proved to be equal to the channel rate and independent of the report and switchover overhead. Throughput performance of backlogged flows depends critically on the service quantum size. A large value (i.e., several packets) is preferable as long as the transport protocol is able to maintain a sufficiently large backlog. The analysis allows an evaluation of the tradeoff between the capacity lost due to transmitter blocking and the cost of equipping sources with more transmitters. 

The present work clearly represents only a preliminary evaluation. It remains to completely specify the way grants are communicated from destination to source. We have only considered an artificially symmetric network configuration. The technological feasibility of 
realizing the supposed burst formation scheme has not been fully explored. Despite these limitations, we believe the present proposal has been shown to hold considerable promise for the development of a new type of cost effective, energy efficient metropolitan area network.  




 

 \appendix{Stability of a flow-aware polling system}
This appendix presents a proof of Proposition \ref{prop:capacity} under the assumption that all flows are backlogged. To be concise and avoid complicated notation, we modify the system somewhat. It is considered as a polling system with a particular service discipline. We assume reports giving the current number of flows in progress are issued for all $S$ sources at the same instant, once per polling cycle. The server then visits each source and serves, for one unit of time, each reported flow before moving to the next source. Reports are issued just before the service of source 1 and take account of terminations and new arrivals at each source since the last report. Before leaving source $i$, after serving the last flow if any, the server remains for a switch overhead of $x_i$ time units. Let $x =x_1+\cdots+x_S$. 

Denote the arrival rate at source $i$ by $\lambda_i$ and its mean flow size and variance by $m_i$ and $v_i$, respectively, with $\rho = \sum \lambda_i m_i$. By assumption, $v_i$ is finite. For convenience we assume $x$ and all flow sizes are integer numbers of time units.  Denote the state of source $i$ by the vector
$N_i=(N_{i,p}, p\geq  1)$ where $n_{i,p}$ is  the remaining size in time units of  the $p^{th}$ flow,
with the convention that $n_{i,p}=0$ if the number of flows is less than $p$. 

$N=(N_i,1\leq i\leq S)$ is the overall system state and we denote by  
$L_i(N)=\sum_{p\geq 1} \ind{N_{i,p}>0}$, the number of flows at node $i$ with $L(N)=\sum_{i=1}^S L_i(N)$. 

Let $N(t)$ be the state at integer valued time $t$ and let $P(t)$ designate the position of the server at this time. $P(t)$ specifies all that is necessary to make $(N(t),P(t))$ an irreducible Markov chain, i.e.,  the node, the flow just served or the remaining overhead before moving to the next node.  

To prove Proposition \ref{prop:capacity}, we show that $(N(t),P(t))$ is ergodic by applying Filonov's theorem \cite[Theorem 8.6]{Robert:08}. To do so, we show that the following is a Lyapunov function for the system: 
\[
\|N\|=\sum_{i=1}^S \sum_{p=1}^{L_i(N)} N_{i,p}^2+\alpha N_{i,p},
\]
where $\alpha$ is some positive constant to be determined.


Assume the server moves to node 1 at time 0 and that the state of the system is $N$. Let $\tau$ denote the ensuing cycle time,
\[
\tau=x+\sum_{i=1}^S L_i(N)=x+L(N).
\]
This is clearly a stopping time with respect to the Markov
chain $(N(t),P(t))$.

At time $\tau$, one time unit of all  flows present at $0$ has been transmitted, i.e.,  for $p \leq L_i(N)$, $N_{i,p}(\tau)=N_{i,p}-1$, with the convention that flows of length
$0$ are removed. Additionally, a number of new flows will have arrived.  We have, therefore,
\begin{align}
\E_N&(\| N(\tau)\|)-\|N\|=\alpha ((\rho-1)L(N)+\rho x)\notag \\
& -2\sum_{i=1}^S\sum_{p=1}^{L_i(N)} N_{i,p}
+L(N)+\E_N(\tau)\sum_{i=1}^S\lambda_i(v_i+m_i^2)\notag\\
&=\alpha (\rho-1)L(N)
-2\sum_{i=1}^S\sum_{p=1}^{L_i(N)} N_{i,p}+ CL(N)+Dx,\label{eq1}
\end{align}
for some constants $C$ and $D$.

Fix  $\alpha$ such that $\alpha(\rho-1)/2+C<0$ and 
$T_0$ such that 
\begin{equation}\label{eq2}
\left(\frac{\alpha}{2}(\rho-1)+C\right)T_0+\left(\frac{\alpha}{2}(1-\rho)+D\right)x\leq 0.
\end{equation}
First assume $L(N)>T_0$. Relations~\eqref{eq1} and~\eqref{eq2} then imply
\begin{equation}
\label{F1}
\E_N(\| N(\tau)\|)-\|N\|\leq  -\frac{\alpha}{2}(1-\rho)\E_N(\tau).
\end{equation}
Now assume $L(N) \leq T_0$. For any positive integer $T_1$, we have
\[
\|N\|>T_1^2+\alpha T_1 \Rightarrow \sum_{i=1}^S\sum_{p=1}^{L_i(N)} N_{k,p}\geq T_1.
\]
Setting $T_1=\lfloor (D+1)x+CT_0\rfloor+1$, we deduce from (\ref{eq1}),
\begin{align}
\E_N(\| N(\tau)\|)-\|N\|&\leq 
-2\sum_{i=1}^S\sum_{p=1}^{L_i(N)} N_{i,p}-x+ CT_0+(D+1)x \notag\\
&\leq -\sum_{i=1}^S\sum_{p=1}^{L_i(N)} N_{i,p}-x \leq -\E_N(\tau).
\label{F2}
\end{align}
We conclude from (\ref{F1}) and (\ref{F2}), that there indeed exists $\gamma>0$ and a stopping time $\tau$ such that, if $\|N\|>T_1^2+\alpha T_1$, then
\[
\E_N(\| N(\tau)\|)-\|N\|\leq -\gamma \E_N(\tau).
\]
This completes the proof under the stated simplifying assumptions. It is possible to remove these assumptions but at the cost of considerably more complicated notation. 

\providecommand{\bysame}{\leavevmode\hbox to3em{\hrulefill}\thinspace}
\providecommand{\MR}{\relax\ifhmode\unskip\space\fi MR }
\providecommand{\MRhref}[2]{%
  \href{http://www.ams.org/mathscinet-getitem?mr=#1}{#2}
}
\providecommand{\href}[2]{#2}

\end{document}